\documentclass[10pt]{article}
\usepackage{amsmath, latexsym, amssymb, amscd, amsfonts, amsthm, epsfig, mathrsfs}
\usepackage[english]{babel}

\DeclareFixedFont{\sfracFont}{U}{euf}{b}{n}{7pt}
\newtheoremstyle{mydefi}
  {15pt}
  {15pt}
  {}
  {}
  {\bfseries}
  {:}
  {.5em}
  {}
\newtheoremstyle{mytheo}
  {15pt}
  {15pt}
  {\slshape}
  {}
  {\bfseries}
  {:}
  {.5em}
  {}

\theoremstyle{mytheo}
\newtheorem{theorem}{Theorem}[section]
\newtheorem{lemma}{Lemma}
\newtheorem{proposition}{Proposition}

\theoremstyle{mydefi}
\newtheorem{definition}{Definition}
\newtheorem{example}{Example}

\newtheorem{remark}{Remark}

\begin{document}
\title{Strong squeezing limit in quantum stochastic models}
\author{Luc Bouten} \date{}
\maketitle

\begin{abstract}
In this paper we study quantum stochastic differential 
equations (QSDEs) \cite{HuP84} that are driven by strongly squeezed 
vacuum noise. We show that for strong squeezing such a QSDE
can be approximated (via a limit in the strong sense) by a QSDE that 
is driven by a single commuting noise process. 
We find that the approximation has an additional Hamiltonian term.   
\end{abstract}

\section{Introduction}\label{sec introduction}

Quantum stochastic differential equations (QSDEs) \cite{HuP84} 
arise via a weak coupling limit from QED \cite{AFLu90, Gou05} and 
are an important tool for modeling the time 
evolution of systems 
that interact with the electromagnetic field in a Markovian approximation
(i.e.\ quantum optics). Many techniques have 
been developed for models that are based on unitaries 
that are given by a QSDE, e.g.\ quantum filtering 
\cite{Bel92b, BHJ07}, adiabatic elimination 
\cite{BoS08, BHS08, GoV07} and input-output
theory \cite{GC84}. 
 
In this paper we will look at QSDEs that are 
driven by squeezed noise \cite{GaZ00, HHKKR}. 
The electromagnetic field acts on a system via 
two field quadrature processes that are both commuting 
processes (i.e.\ given one of these processes: the 
operators at different times of the process commute with 
each other), but the two quadrature processes do not 
commute with each other. In terms of the Hudson-Parthasarathy
theory these quadrature processes are given by linear 
combinations of the annihilator process $A_t$ and creation 
process $A_t^*$, namely $A_t+A_t^*$ and $i(A_t-A_t^*)$. 
With respect to the vacuum state both of these quadrature process 
are Wiener processes (but these Wiener processes do not 
commute with each other). In the squeezed vacuum one of 
the processes has an increased variance and the other 
has a decreased variance. 

In the case of strong squeezing we expect that we can 
neglect the noise with the small variance and can approximate 
the system as if it was driven only by the noise with the 
large variance. In this paper we are going to make this 
idea precise. As we will see in Theorem \ref{thm main result}: 
in case of strong squeezing 
the time evolution can be approximated well by an equation 
that is driven only by one commutative noise process. 
We call such dynamics \emph{essentially commutative} \cite{KuM87}.  
As observed in \cite{GrW02}, in principle it is possible to 
completely undo the decoherence for essentially commutative 
dynamics. This was studied on the level of the 
filter in \cite[Chapter 4]{Bou04b}, 
where a control scheme was introduced that restores 
quantum information (i.e.\ completely freezes the time 
evolution of the filter estimates).

We prove that the difference of the essentially commutative 
approximation and the original system dynamics converges 
strongly to zero. The proof is heavily inspired by the proof 
of the adiabatic theorem in \cite{BoS08} and relies heavily 
on the Trotter-Kato Theorem \cite{Tro58,Kat59}. The essentially 
commutative approximation has an additional Hamiltonian term 
in its dynamics when compared to the original dynamics. 

The remainder of this article is organized as follows: 
Section \ref{sec main result} introduces the system, 
its essentially commutative approximation and states 
the main theorem (Theorem \ref{thm main result}).
In Section \ref{sec examples} we apply the main 
Thm in example systems. We conclude the article 
with Section \ref{sec proof of main result} in which 
we proof the main theorem.

\section{The main result}\label{sec main result}

Throughout this paper $n$ is a positive 
real number, $c = |c|\exp(i\theta)$ is a 
complex number such that $|c| = \sqrt{n(n+1)}$, 
$a$ denotes the real part of $c$,
$\mathcal{H}$ is a separable Hilbert space 
(the \emph{initial space}) and $\mathcal{F}$ 
is the symmetric Fock space over $L^2(\mathbb{R}^+)$.
We denote the vacuum vector in $\mathcal{F}$ by $\Phi$. 
On the Fock space $\mathcal{F}$ we have the 
usual annihilation process $A_t$, creation 
process $A^*_t$ and gauge process $\Lambda_t$ 
in the sense of Hudson and Parthasarathy \cite{HuP84}. 
These noises satisfy the following quantum It\^o 
table \cite{HuP84}:

\begin{center}
{\begin{tabular} {l|lll}
$$ & $dA^*_t$ & $d\Lambda_t$ & $dA_t$ \\
\hline 
$dA^*_t$ & $0$ & $0$ & $0$ \\
$d\Lambda_t$ & $dA^*_t$ & $d\Lambda_t$ & $0$  \\
$dA_t$ & $dt$ & $dA_t$ & $0$ 
\end{tabular} } 
\end{center}

We define the \emph{squeezed noise processes} 
$B_t$ and $B_t^*$ on $\mathcal{F}$ as the 
following linear combinations 
of $A_t$ and $A_t^*$:
\begin{equation}\begin{split}\label{B}
&B_t  :=
\frac{n+c}{\sqrt{2n+1+2a}}A^*_t + \frac{n+1+c}{\sqrt{2n+1+2a}}A_t,\\
&B^*_t := 
\frac{n+\overline{c}}{\sqrt{2n+1+2a}}A_t + 
  \frac{n+1+\overline{c}}{\sqrt{2n+1+2a}}A_t^*.
\end{split}\end{equation}

Note that these noises obey the \emph{squeezed noise}
quantum It\^o table \cite{GaZ00}
\begin{center}
{\begin{tabular} {l|lll}
$$ & $dB_t^*$ & $dB_t$ \\
\hline 
$dB^*_t$ & $\overline{c}dt$ & $ndt$ \\
$dB_t$ & $(n+1)dt$ & $cdt$ 
\end{tabular} }
\end{center}
Note that if $c=a$ is a real number, then 
$B_t + B_t^* = \sqrt{2n + 1 +2a}(A_t+ A_t^*)$ and 
$i(B_t-B_t^*) = i(A_t-A_t^*)/\sqrt{2n +1 +2a}$, i.e.\ 
with respect to the vacuum 
one is a Wiener process with an increased variance whereas 
the other has a decreased variance.

In this paper we study the following 
quantum stochastic differential equation (QSDE)
on $\mathcal{H}\otimes \mathcal{F}$ 
in the sense of Hudson and Parthasarathy \cite{HuP84}:
\begin{equation}\begin{split}
d\tilde{U}_t^{nc} = \Bigg\{&LdB_t^* - L^*dB_t \  + \\
& + \frac{1}{2}\Big(LL\overline{c}- LL^*n -L^*L(n+1) +L^*L^*c\Big)dt 
- iHdt\Bigg\}\tilde{U}_t^{nc}, \ \ \tilde{U}_0^{nc} = I.  
\end{split}\end{equation}
Here $L$ and $H$ are assumed to be 
bounded operators on $\mathcal{H}$ such 
that $S$ is unitary and $H$ is self-adjoint.
We will make the definition a bit more general 
later and then we will drop de tilde in the notation. 
Note that the solution to the above equation 
is unitary \cite{HuP84}.

We now define:
\begin{equation}\begin{split}\label{L and F}
&L_{nc} := \frac{n+1+\overline{c}}{\sqrt{2n + 1 + 2a}}L - \frac{n+c}{\sqrt{2n+1+2a}}L^*,\\
&F_{nc} := \frac{n+\overline{c}}{\sqrt{2n + 1 + 2a}}L - \frac{n+c}{\sqrt{2n+1+2a}}L^*.
\end{split}\end{equation}
Note that $F_{nc}$ is skew-selfadjoint.

Using the definition of $L_{nc}$  
in Eqn (\ref{L and F}) and the definitions 
of $B_t$ and $B_t^*$ in Eqn (\ref{B}) we 
find after some re-arranging: 
\begin{equation}\label{U no Lambda}
d\tilde{U}_t^{nc} = \left\{L_{nc}dA^*_t - L_{nc}^*dA_t -
\frac{1}{2}L_{nc}^*L_{nc}dt
-iH dt\right\}\tilde{U}_t^{nc}, \ \ \ \ \tilde{U}_0^{nc} = I
\end{equation}
Note that we could re-write Eqn \eqref{U no Lambda} in a way 
that makes it clear that the equation is driven by two in themselves 
commuting noise processes $\{i(A_t -A_t^*)\}_{t\ge 0}$ and 
$\{A_t+A_t^*\}_{g\ge0}$. However, 
these two noises do not commute with each other. 
\begin{equation*}
d\tilde{U}_t^{nc} = \left\{\frac{L_{nc} + L_{nc}^*}{2}(dA^*_t - dA_t) + 
\frac{L_{nc}-L_{nc}^*}{2}(dA_t+dA_t^*) -
\frac{1}{2}L_{nc}^*L_{nc}dt
-iH dt\right\}\tilde{U}_t^{nc}.
\end{equation*}
Note that if $n$ becomes very large (strong squeezing), 
then $1$ is negligible with respect to $n$. This is why 
we expect that for large $n$ we can replace the operators 
$L_{nc}$ by $F_{nc}$. This can significantly reduce the 
complexity of the interaction between the system living 
on $\mathcal{H}$ and the field that lives on $\mathcal{F}$. 
If we replace $L_{nc}$ by $F_{nc}$ in 
Eqn (\ref{U no Lambda}), then we see since $F_{nc}^* = -F_{nc}$,  
that the QSDE is now driven by only one classical noise process 
$\{A_t+A_t^*\}_{t\ge0}$. QSDE's that are driven 
by noises that are commutative in themselves and also all 
commute with each other are called 
\emph{essentially commutative} \cite{KuM87}.

We can generalize Eqn (\ref{U no Lambda}) by introducing 
a gauge term in the equation. Often these terms appear 
after an adiabatic elimination procedure \cite{BoS08, BHS08}.
\begin{equation}\label{U}
dU_t^{nc} = \left\{(S-I)d\Lambda_t + L_{nc}dA^*_t - L_{nc}^*SdA_t -
\frac{1}{2}L_{nc}^*L_{nc}dt
-iH dt\right\}U_t^{nc}, \ \ \ \ U_0^{nc} = I.
\end{equation}
Here $S$ is a unitary operator on $\mathcal{H}$.

We now introduce the following QSDE:
\begin{equation}\begin{split}\label{V}
&dV_t^{nc} = \Big\{(S-I)d\Lambda_t + F_{nc}dA^*_t - F_{nc}^*SdA_t -
\frac{1}{2}F_{nc}^*F_{nc}dt -i(H+H_{nc})dt
\Big\}V_t^{nc}, \\ 
&V_0^{nc} = I.
\end{split}\end{equation}
Here the Hamiltonian $H_{nc}$ is given by
\begin{equation*}
H_{nc} = -\frac{i}{2}\Big(\frac{n+\overline{c}}{2n +1 +2a}L^2 
-\frac{n+c}{2n +1 +2a}{L^*}^2 + 
\frac{(\overline{c}-c)}{2n +1 + 2a}L^*L
\Big).
\end{equation*}

We can now state our main result:
\begin{theorem}\label{thm main result}
Let $U_t^{nc}$ be given by Eqn \eqref{U} and $V_t^{nc}$ by 
Eqn \eqref{V}. We then have  
  \begin{equation*}
  \lim_{n\to \infty} \left\|\big(U_t^{nc} -V_t^{nc}\big) \psi\right\| =0, \ \ \ \
  \forall \psi \in \mathcal{H}\otimes \mathcal{F}. 
  \end{equation*}   
\end{theorem}
Note that $|c| = \sqrt{n(n+1)}$ also goes 
to inifity as $n$ goes to infinity. The phase 
$\theta$ of $c = |c|\exp(i\theta)$ stays constant.

\begin{proof}
See Section \ref{sec proof of main result}.
\end{proof}

\begin{remark}
If one studies the proof of Theorem \ref{thm main result} 
in Section \ref{sec proof of main result}, then one easily 
sees that the Theorem could be stated a little bit more 
general. It is possible to add in extra channels that 
do not scale with $n$, provided that they are are 
present both in Eqn \eqref{U} and Eqn \eqref{V} in 
the same way. It is even possible to have scattering 
$S_{ij}$ between the channel that does scale with n 
and the other channels. We have not stated the Theorem 
in this way, because it is an obvious generalization and 
it would force us to carry a lot of notation around.
\end{remark}

\begin{remark}
Define
\begin{equation}\begin{split}\label{Z}
&Z_t  :=
\frac{n+c}{\sqrt{2n+1+2a}}A^*_t + \frac{n+c}{\sqrt{2n+1+2a}}A_t,\\
&Z^*_t := 
\frac{n+\overline{c}}{\sqrt{2n+1+2a}}A_t + 
  \frac{n+\overline{c}}{\sqrt{2n+1+2a}}A_t^*.
\end{split}\end{equation}
Note that it immediately follows that 
these noises satisfy the following quantum 
It\^o table
\begin{center}
{\begin{tabular} {l|lll}
$$ & $dZ_t^*$ & $dZ_t$ \\
\hline 
$dZ^*_t$ & $\left(\overline{c} - \frac{n+\overline{c}}{2n+1+2a}\right)dt$ & $ndt$ \\
$dZ_t$ & $ndt$ & $\left(c - \frac{n+c}{2n+1+2a}\right)dt$ 
\end{tabular} }
\end{center}
Suppose that $S=I$ in Eqn \eqref{U} and Eqn \eqref{V}. Rewriting
Eqn \eqref{U} in terms of the noises $B_t$ and $B_t^*$ and 
Eqn \eqref{V} in terms of the noises $Z_t$ and $Z_t^*$, we find
\begin{equation*}\begin{split}
dU_t^{nc} = \Bigg\{&LdB_t^* - L^*dB_t + 
\frac{1}{2}\Big(LL\overline{c}- LL^*n -L^*L(n+1)\ + \\ 
&+L^*L^*c\Big)dt 
- iHdt\Bigg\}U_t^{nc}, \ \ \ \ U_0^{nc} = I,\\
dV_t^{nc} = \Bigg\{&LdZ_t^* - L^*dZ_t 
+ \frac{1}{2}\Bigg(LL\left(\overline{c} - \frac{n+\overline{c}}{2n+1+2a}\right) 
- LL^*n -L^*Ln\ + \\
&+ L^*L^*\left(c - \frac{n+c}{2n+1+2a}\right)\Bigg)dt 
- i(H+H_{nc})dt\Bigg\}V_t^{nc}, \ \ \ \ V_0^{nc} = I.
\end{split}\end{equation*}
This provides a second perspective on Thm \ref{thm main result}: 
instead of replacing the coefficients $L_{nc}$ by $F_{nc}$, we 
can equivalently replace the noises $B_t,B_t^*$ by $Z_t, Z_t^*$ 
(where in both cases we also have to add 
the extra Hamiltonian term $H_{nc}$).
Both procedures lead to the same approximation 
for the case of strong squeezing. 
\end{remark}

\section{Examples}\label{sec examples}

\begin{example}
{\bf (Two level atom coupled to strongly squeezed noise)} 
Let $\sigma_+$ and $\sigma_-$ be the usual 
two level raising and lowering operators
\begin{equation*}
\sigma_+ = \begin{pmatrix}0 & 1 \\ 0 & 0\end{pmatrix}, \ \ \ \ 
\sigma_- = \begin{pmatrix}0 & 0 \\ 1 & 0\end{pmatrix}.
\end{equation*}
A two-level atom driven by squeezed light can be 
described by the following QSDE
\begin{equation}
dU_t = \left\{\kappa\sigma_{nc}dA_t^* -\kappa\sigma_{nc}^* dA_t 
-\frac{1}{2}\kappa^2\sigma_{nc}^*\sigma_{nc}dt -iHdt\right\}U_t,\ \ \ \ U_0= I. 
\end{equation}
Here $\kappa^2$ is the decay rate of the two-level atom, 
$H$ is an internal atom Hamiltonian and $\sigma_{nc}$ is 
given by
\begin{equation*}
\sigma_{nc} = \frac{n+1+\overline{c}}{\sqrt{2n + 1 + 2a}}\sigma_- 
- \frac{n+c}{\sqrt{2n+1+2a}}\sigma_+.
\end{equation*}

This system was studied \cite[Chapter 4]{Bou04b} 
in the strong squeezing limit 
at the level of the quantum filter (see \cite{BHJ07} for a review 
of quantum filtering theory). The aim was to control the 
decoherence. It turns out that 
with the control strategy proposed in \cite{Bou04b} 
it is possible to freeze the system dynamics. That is: the 
estimates from the filter have no time evolution any more. 
The reason why the control strategy works, is because the 
system dynamics become essentially 
commutative \cite{KuM87} in the strong squeezing limit.  

Theorem \ref{thm main result} shows that it 
is possible to approximate the system 
already at the level of the 
unitary evolution from which the filter needs to be derived.
It follows from Thm \ref{thm main result} that in 
the case of strong squeezing (large $n$), the system 
can be approximated by the following unitary 
evolution
\begin{equation*}\begin{split}
&V_t^{nc} = \left\{\kappa\gamma_{nc}(dA_t^* + dA_t) 
+\frac{1}{2}\kappa^2\gamma_{nc}^2dt -iHdt -iH_{nc}dt\right\}V_t^{nc}, 
\ \ \ \  V_0=I,\\
&\gamma_{nc} = \frac{n+\overline{c}}{\sqrt{2n + 1 + 2a}}\sigma_- 
- \frac{n+c}{\sqrt{2n+1+2a}}\sigma_+,\\
&H_{nc} = -\frac{i}{2} \frac{(\overline{c}-c)\kappa^2}{2n+1+2a} \sigma_+\sigma_-.
\end{split}\end{equation*} 
This equation is indeed only driven by one commuting 
noise process: $A_t+A_t^*$. 
That is: there is no $i(A_t-A_t^*)$ term driving $V_t^{nc}$: the 
dynamics is essentially commutative.  
\end{example}

\begin{example}
{\bf (A cavity coupled to strongly squeezed noise)}
We consider a cavity coupled to squeezed vacuum noise 
via one of its mirrors. The system lives on the Hilbert 
space $\ell^2(\mathbb{N})\otimes\mathcal{F}$ and is given 
by
\begin{equation*}\begin{split}
&dU_t^{nc} = \left\{\kappa b_{nc}dA_t^* -\kappa b_{nc}^*dA_t 
-\frac{1}{2}\kappa^2b_{nc}^*b_{nc} -i\hbar \omega b^*b\right\}U_t^{nc}, 
\ \ \ \ U_0 = I,\\
&b_{nc} = \frac{n+ 1 + \overline{c}}{\sqrt{2n + 1 + 2a}}b 
- \frac{n+c}{\sqrt{2n+1+2a}}b^*.
\end{split}\end{equation*} 
\end{example}
Here $\kappa^2$ is the decay rate of the cavity, $\omega$ 
is the cavity frequency and $b$ is the standard lowering 
operator and $b^*$ is the standard raising operator for the 
eigen functions of $b^*b$
\begin{equation*}\begin{split}
b\phi_i = \sqrt{i} \phi_{i-1},\ \ \  b^*\phi_i = \sqrt{i+1}\phi_{i+1}, \ \ \ 
b^*b \phi_i = i\phi_i.
\end{split}\end{equation*}
Note that $[b,b^*] = 1$. The operators $b$ and $b^*$ 
are unbounded which means that this example is technically 
out of the scope of Theorem \ref{thm main result}. We fix this by 
simply truncating the operators at a very high level $N$.

We can now apply Theorem \ref{thm main result} and 
find that the time evolution of the cavity and its 
environment in the case of strong squeezing can 
be approximated by
\begin{equation*}\begin{split}
&dV_t^{nc} = \left\{\kappa f_{nc}(dA_t^* + dA_t) 
+\frac{1}{2}\kappa^2f_{nc}^2 -i(\hbar \omega b^*b + H_{nc})\right\}V_t^{nc}, 
\ \ \ \ V_0 = I,\\
&f_{nc} = \frac{n+ \overline{c}}{\sqrt{2n + 1 + 2a}}b 
- \frac{n+c}{\sqrt{2n+1+2a}}b^*, \\
&H_{nc} = -\frac{i\kappa^2}{2}\Big(\frac{n+\overline{c}}{2n +1 +2a}b^2 
-\frac{n+c}{2n +1 +2a}{b^*}^2 + 
\frac{(\overline{c}-c)}{2n +1 + 2a}b^*b
\Big).
\end{split}\end{equation*} 
Notice that the dynamics given by 
$V_t^{nc}$ is again essentially commutative.

\section{Proof of Theorem \ref{thm main result}}\label{sec proof of main result}

Let $\alpha$ be a complex number. We define the Weyl 
operator $W_t(\alpha)$ as the solution to the following
QSDE
\begin{equation*}
dW_t(\alpha) = \left\{\alpha dA_t^* - \overline{\alpha}dA_t 
-\frac{1}{2}|\alpha|^2dt \right\}W_t(\alpha), \ \ \ W_0(\alpha) = I.
\end{equation*}
Now we define 
$U_t^{nc}(\alpha) := U_t^{nc}W_t(\alpha)$ and 
$V_t^{nc}(\alpha) := V_t^{nc}W_t(\alpha)$. It then follows 
from the quantum It\^o rules that
\begin{equation}\begin{split}\label{generator V}
dV_t^{nc}(\alpha)^* = V_t^{nc}(\alpha)^*&
\Big\{(S^*-I)d\Lambda_t +(F_{nc}+\alpha)^*dA_t  - S^*(F_{nc}+{\alpha})dA^*_t \\ 
& -\frac{1}{2}(F_{nc}+\alpha)^*(F_{nc}+\alpha) dt  \\
& - \frac{1}{2}(\overline{\alpha}F_{nc} - \alpha F_{nc}^*)dt+i(H+H_{nc})dt\Big\}, 
\end{split}\end{equation}
\begin{equation}\begin{split}\label{generator U}
dU_t^{nc}(\beta) = \Big\{&(S-I)d\Lambda_t + (L_{nc}+\alpha)dA_t^*  - (L_{nc}+\alpha)^*SdA_t \\ 
& -\frac{1}{2}(L_{nc}+\alpha)^*(L_{nc}+\alpha) dt \ + \\
& + \frac{1}{2}(\overline{\alpha}L_{nc} -\alpha L_{nc}^*)dt-iHdt\Big\}U_t^{nc}(\alpha). 
\end{split}\end{equation}

\begin{definition}\label{semigroups}
We denote by $\mbox{id}: \mathcal{B}(\mathcal{H}) \to \mathcal{B}(\mathcal{H})$  
the identity map $\mbox{id}(X) = X$. We 
write $\phi$ for the state on $\mathcal{B}(\mathcal{H})$ given by taking the 
inner product with the vacuum vector $\Phi$. We let $\mathcal{B}_0$ be the 
Banach subalgebra of $\mathcal{B}(\mathcal{H})$ generated by the 
identity element $I$ in $\mathcal{B}(\mathcal{H})$. We now 
define: 
\begin{equation*}\begin{split}
&T_t^{(\alpha n c)}(X) := \mbox{id}\otimes \phi\Big(
V_t^{nc}(\alpha)^*XU_t^{nc}(\alpha)
\Big), \ \ \mbox{for all} \ \
X \in \mathcal{B}(\mathcal{H}),\ t\ge 0,\\
&T_t(X) := X,\ \ \mbox{for all} \ \ X \in \mathcal{B}_0,\ t\ge 0.
\end{split}\end{equation*} 
\end{definition}

\begin{lemma}
For every $n\ge 0$, the families of bounded linear maps 
$T_t^{(\alpha n c)}$ $t\ge 0$ and $T_t$ $(t\ge0)$ given by 
Definition \ref{semigroups} are norm-continuous one-parameter 
semigroups. 
\end{lemma}

\begin{proof}
The semigroup property of $T_t^{(\alpha n c)}$ follows immediately 
from the cocycle property (wrt the shift) of $V_t^{nc}(\alpha)$
and $U_t^{nc}(\alpha)$. Since the conditional 
expectation $\mbox{id}\otimes\phi$ is norm-contractive and 
$V^{nc}_t(\alpha)$ and $U^{nc}_t(\alpha)$ are unitary, we have
\begin{equation*}
\left\| T^{(\alpha n c)}_t(X)\right\| \le 
\left\|{V^{nc}_t(\alpha)}^* X U^{nc}_t(\alpha)\right\| \le
\left\|{V^{nc}_t(\alpha)}^*\right\| \left\| X\right\| 
\left\| U^{nc}_t(\alpha)\right\| \le \| X\|,
\end{equation*} 
i.e.\ $T_t^{(\alpha nc)}$ is norm-contractive. Note that 
due to the boundedness of all coefficients in the QSDEs 
for $V^{nc}_t(\alpha)$ and $U_t^{nc}(\alpha)$ 
(Eqns \eqref{generator V} and \eqref{generator U}), it immediately follows that 
the generator of $T_t^{(\alpha n c)}$ is bounded. This means 
that $T_t^{(\alpha n c)}$ is norm-continuous. Note that 
the statements about $T_t$ are trivially true.
\end{proof}

\begin{proposition}\label{propostion generator}
The generator $\mathscr{L}^{(\alpha n c)}$ of the semigroup
$T_t^{(\alpha nc)} = \exp(t\mathscr{L}^{(\alpha n c)})$ $(t\ge 0)$
evaluated at the identity element $I$ of $\mathcal{B}(\mathcal{H})$
is given by
\begin{equation}\begin{split}\label{generator}
\mathscr{L}^{(\alpha n c)}(I) &= -\frac{1}{2}(L_{nc}-F_{nc})^*(L_{nc}-F_{nc})
+ \overline{\alpha}(L_{nc}-F_{nc}) -\alpha (L_{nc}-F_{nc})^*\ = \\
&= -\frac{1}{2}\frac{L^*L}{2n+1+2a} + \frac{\overline{\alpha}L -\alpha L^*}{\sqrt{2n+1+2a}}.
\end{split}\end{equation}
\end{proposition}

\begin{proof}
Note that $dT_t^{\alpha nc}(I) = 
\mbox{id}\otimes \phi (d({V_t^{nc}(\alpha)}^*U_t^{nc}(\alpha))) = 
T_t^{(\alpha nc)}(\mathscr{L}^{(\alpha nc)}(I))dt$. 
Using Eqns \eqref{generator V} and \eqref{generator U}, the 
quantum It\^o rule \cite{HuP84} and the fact that vacuum expectations of 
stochastic integrals vanish \cite{HuP84}, we find
\begin{equation*}\begin{split}
\mathscr{L}^{(\alpha n c)}(I) ={ } &-\frac{1}{2}(F_{nc}+\alpha)^*(F_{nc}+\alpha)   
- \frac{1}{2}(\overline{\alpha}F_{nc} - \alpha F_{nc}^*) + i(H+H_{nc}) \\
&-\frac{1}{2}(L_{nc}+\alpha)^*(L_{nc}+\alpha)  
+ \frac{1}{2}(\overline{\alpha}L_{nc} -\alpha L_{nc}^*)-iH\\
&+ (F_{nc} + \alpha)^*(L_{nc} + \alpha).
\end{split}\end{equation*}
We can easily re-write this to obtain
\begin{equation*}\begin{split}
\mathscr{L}^{(\alpha n c)}(I) ={ } &-\frac{1}{2}\Big(
(L_{nc}+\alpha)^*(L_{nc}+\alpha) 
+(F_{nc}+\alpha)^*(F_{nc}+\alpha)
- 2(F_{nc} + \alpha)^*(L_{nc} + \alpha)\Big)  
\\
&  
+ \frac{1}{2}(\overline{\alpha}(L_{nc}-F_{nc}) -\alpha (L_{nc}-F_{nc})^*) 
+ iH_{nc}. 
\end{split}\end{equation*}
Now we complete the squares and obtain
\begin{equation*}\begin{split}
\mathscr{L}^{(\alpha n c)}(I) ={ } &-\frac{1}{2}\Big(
(L_{nc}-F_{nc})^*(L_{nc}-F_{nc}) 
- (F_{nc} + \alpha)^*(L_{nc} + \alpha)
+(L_{nc}+\alpha)^*(F_{nc}+\alpha)\Big)  
\\
&  
+ \frac{1}{2}(\overline{\alpha}(L_{nc}-F_{nc}) -\alpha (L_{nc}-F_{nc})^*) 
+ iH_{nc}. 
\end{split}\end{equation*}
Taking all $\alpha$ and $\overline{\alpha}$ terms 
together, we find
\begin{equation*}\begin{split}
\mathscr{L}^{(\alpha n c)}(I) ={ } &-\frac{1}{2}
(L_{nc}-F_{nc})^*(L_{nc}-F_{nc})
+ \overline{\alpha}(L_{nc}-F_{nc}) -\alpha (L_{nc}-F_{nc})^*  
\\
&  
+ \frac{1}{2}(F_{nc}^*L_{nc}- L_{nc}^*F_{nc}) 
+ iH_{nc}. 
\end{split}\end{equation*}
The proposition now follows from the definition of 
$L_{nc}, F_{nc}$ and $H_{nc}$.
\end{proof}

The proof of our main result (Theorem \ref{thm main result})
relies heavily on the Trotter-Kato theorem \cite{Tro58,Kat59}. 
We have taken the formulation 
of the Trotter-Kato theorem 
from \cite[Thm 3.17, page 80]{Dav80}. 

\begin{theorem}{\bf Trotter-Kato Theorem}\label{thm Trotter-Kato}
Let $\mathcal{B}$ be a Banach space and let 
$\mathcal{B}_0$ be a closed subspace of $\mathcal{B}$. 
For each $n \ge 0$, let $T_t^{(n)}$ be a strongly 
continuous one-parameter contraction semigroup on 
$\mathcal{B}$ with generator 
$\mathscr{L}^{(n)}$. Moreover, let $T_t$ be a 
strongly continuous one-parameter contraction 
semigroup on $\mathcal{B}_0$ with generator $\mathscr{L}$. 
Let $\mathcal{D}$ be a core for $\mathscr{L}$.
The following conditions are equivalent:
  \begin{enumerate}
  \item For all $X \in \mathcal{D}$ there exist 
  $X^{(n)} \in \mbox{Dom}\left(\mathscr{L}^{(n)}\right)$ such 
  that 
    \begin{equation*}
    \lim_{n \to \infty} X^{(n)} = X, \qquad 
    \lim_{n \to \infty} \mathscr{L}^{(n)}\left(X^{(n)}\right) = 
    \mathscr{L}(X).
    \end{equation*}
  \item For all $0 \le s < \infty$ and all $X \in \mathcal{B}_0$ 
    \begin{equation*}
    \lim_{n \to \infty} \left\{\sup_{0 \le t \le s} 
    \left\|T_t^{(n)}(X) - T_t(X)\right\|\right\} = 0.
    \end{equation*}   
  \end{enumerate}
\end{theorem}

\begin{proposition}\label{proposition generator}
Let $T_t^{(\alpha nc)}$ and $T_t$ be the one-parameter 
semigroups on $\mathcal{B}(\mathcal{H})$ and $\mathcal{B}_0$
defined in Definition \ref{semigroups}, respectively. 
We now have:
  \begin{equation*}
      \lim_{n \to \infty} \left\{\sup_{0 \le t \le s} 
    \left\|T_t^{(\alpha nc)}(I) - I\right\|\right\} = 0,
  \end{equation*} 
for all $0 \le s < \infty$.  
\end{proposition}

\begin{proof}
Note that both semigroups are norm-continuous and therefore
also strongly continuous. Note that the generator $\mathscr{L}$
of $T_t = \exp(t\mathscr{L})$ is equal to $0$. We are now 
going to apply the Trotter-Kato theorem (Thm \ref{thm Trotter-Kato})
with $\mathcal{D} = \mathcal{B}_0$. If we take $X^{(n)} =I$ 
for all $n$, then obviously we have $\lim_{n \to \infty } X^{(n)} = I$. 
It follows from Propostion \ref{propostion generator} that 
\begin{equation*}
\lim_{n\to \infty} \mathscr{L}^{(\alpha nc)}(X^{(n)}) = 
\lim_{n\to \infty} \mathscr{L}^{(\alpha nc)}(I) = 0 = \mathscr{L}(I).
\end{equation*}
Since all elements in $\mathcal{B}_0$ are multiples of $I$, we have 
the above result for all elements in $\mathcal{B}_0$. The proposition 
then follows from the Trotter-Kato Theorem and the fact 
that $T_t(I) =I$.
\end{proof}

Let $f$ be a function in $L^2(\mathbb{R})$. We define
the Weyl operator $W_t(f) = W(f\chi_{[0,t]})$ (where 
$\chi_{[0,t]}$ is the indicator function of the interval 
$[0,t]$), by the following QSDE
\begin{equation*}
dW_t(f) = \left\{f(t)dA_t^* - \overline{f(t)}dA_t 
-\frac{1}{2}|f(t)|^2dt\right\}W_t(f), \ \ \ \ W_0(f) = I.
\end{equation*}
If we act with W(f) on the vacuum $\Phi$, then we get 
the coherent vector $\psi(f)$. The coherent vectors form 
a dense set in $\mathcal{F}$.

\noindent{\bf Proof of Theorem \ref{thm main result}}
Let $t \ge 0$. Let $f$ be a step function in 
$L^2([0,t])$, i.e.\ there 
exists an $m \in \mathbb{N}$ and $0 =t_0 < t_1<\ldots<t_m =t$
and $\alpha_1,\ldots,\alpha_m \in \mathbb{C}$ such 
that 
  \begin{equation*}
  s \in [t_{i-1},t_i) \Longrightarrow f(s) = \alpha_i,\qquad 
  \forall i \in \{1,\ldots,m\}.
  \end{equation*}
Let $\psi(f)$ be the coherent vector with respect to $f$. 
Let $v$ be an element in $\mathcal{H}$. 
The cocycle property of solutions to QSDE's and 
the exponential property of the symmetric Fock 
space lead to 
\begin{equation*}\begin{split}
\Big\langle v\otimes \psi(f),\, 
{V_t^{nc}}^*U_t^{nc}v\otimes \psi(f)\Big\rangle &=
\Big\langle v\otimes \Phi,\, 
\big(V_t^{nc}W(f)\big)^*U_t^{nc}W(f)v\otimes \Phi\Big\rangle = \\
&= \Big\langle v,\, T_{t_1}^{(\alpha_1 nc)}  \cdots 
T_{t-t_m}^{(\alpha_m nc)}(I) v\Big\rangle
\end{split}\end{equation*}
Now we have due to Proposition \ref{proposition generator}
\begin{equation*}\begin{split}
&\lim_{n\to \infty}\big\|(U_t^{nc}-V_t^{nc})v\otimes\psi(f)\big\|^2 = \\
&\lim_{n\to \infty} \Bigg\langle v,\, \Big(2I - T_{t_1}^{(\alpha_1 nc)}  \cdots 
T_{t-t_m}^{(\alpha_m nc)}(I) - T_{t_1}^{(\alpha_1 nc)}  \cdots 
T_{t-t_m}^{(\alpha_m nc)}(I)^*\Big)v\Bigg\rangle =0.
\end{split}\end{equation*}
The Thm now follows because the step functions are 
dense in $L^2(\mathbb{R})$ and the span of all 
coherent vectors, i.e.\ $\mbox{span}\{\Psi(f), f\in L^2(\mathbb{R})\}$,  
is dense in $\mathcal{F}$. \qed

\bibliography{squeez}
\end{document}